\documentclass[12pt,a4paper]{article}
\usepackage{amssymb}
\usepackage{amsmath,amssymb}
\usepackage[all]{xy}
 \usepackage{latexsym}
\usepackage{amsthm,a4wide}
\input diagxy

  \theoremstyle{plain}

  \newtheorem{thm}{Theorem}[section]
  \newtheorem{lem}[thm]{Lemma}

  \theoremstyle{definition}

\begin{document}
\newcommand{\bm}{\bibitem}

\title{{\ New Periodic Solutions for Some Planar \\$N+3$-Body Problems with Newtonian Potentials \thanks{Supported
partially by NSF of China.}}}
\author{ {\normalsize   Pengfei Yuan and Shiqing Zhang}\\
{\normalsize pfyuan123@gmail.com,\quad  zhangshiqing@msn.com}\\
{\small Department of Mathematics, Sichuan University, Chengdu
610064, China}}
\date{}
\maketitle
\begin{abstract}For some planar Newtonian $N+3$-body problems, we use variational minimization methods to prove the existence of new  periodic solutions satisfying that $N$ bodies chase each other on a curve, and the other $3$  bodies chase each other on another curve. From the definition of the group action in equations $(3.1)-(3.3)$, we can find that they are new solutions which are also different from  all the examples of Ferrario and Terracini (2004)$[22]$.\\[5pt]
\\[4pt]
\bf{2000 Mathematicals Subject Classification: 34C15, 34C25, 58F}
\end{abstract}
%
\section{Introduction and  Main Results}

 In recent years, many authors used methods of minimizing the Lagrangian action  on a symmetric space to study the periodic solutions for Newtonian $N$-body problem $([2],[4]-[6],[8]-[29],[31]-[40])$. Especially, A.Chenciner-R.Montgomery $[16]$ proved the existence of the remarkable figure eight type periodic solution for Newtonian three-body problem with equal masses, C.Sim\'{o} $[32]$  discovered many new periodic solutions for Newtonian $N$-body problem using numerical methods. C.Machal $[27]$ studied the fixed-ends (Bolza) problem for Newtonian $N$-body problem and proved that the minimizer for the Lagrangian action has no interior collision; A.Chenciner $[12]$, D.Ferario and S.Terracini $[22]$ simplified and developed C.Marchal's important works; S.Q.Zhang $[36]$, S.Q.Zhang, Q.Zhou $([37]-[40])$ decomposed the Lagrangian action  for $N$-body problem into some sum for two-body problem and compared the lower bound for the lagrangian action on  test orbits  with the upper bound on collision set to avoid collisions under some cases.  Motivated  by the works of A.Chenciner and R.Montgomery, C.Sim\'{o}, C.Marchal, S.Q.Zhang and Q.Zhou, K.C. Chen $([8]-[11])$ studied some planar $N$-body problems  and got some new planar non-collision periodic and quasi-periodic solutions.

 The equations for the motion of the  Newtonian $N$-body problem are:

\begin{align}
m_{i}\ddot{q}_i=\frac{\partial U(q)}{\partial q_i}, \quad  i=1,\ldots, N,   \tag{1.1}
\end{align}

where $q_i\in \mathbb{R}^k$  denotes the position of $m_i$,  and the potential function is :
$$
U=\sum_{1\leq i< j \leq N}^{}\dfrac{ m_i m_j}{|q_i-q_j|}.
$$

It is well known that  critical points of the action functional $f$:
$$
f(q)=\int_0^T(\frac{1}{2}\sum_{i=1}^{N}m_i|\dot{q}_i|^2+U(q))dt,\quad q\in E,\eqno{(1.2)}
$$

are $T$ periodic solutions of the $N$-body problem  $(1.1)$,

where
$$
E=\{ q=(q_1,q_2,\ldots, q_{N}) \,|\, q_i(t)\in W^{1,2}(\mathbb{R}/T\mathbb{Z},\mathbb{R}^k),\,\sum_{i=1}^N m_i q_i(t)=0,\,q_i(t)\neq q_j(t), \forall i\neq j, \forall t\in \mathbb{R} \} ,\eqno{(1.3)}
$$

$$
W^{1,2}(\mathbb{R}/T\mathbb{Z},\mathbb{R}^k)=\{x(t)\,|\,x(t)\in L^2(\mathbb{R},\mathbb{R}^k),\dot{x}(t)\in L^2(\mathbb{R},\mathbb{R}^k) ,x(t+T)=x(t) \}  .\eqno{(1.4)}
$$

{ \bf Definition 1.1}\quad {\it Let $\Gamma: x(t),\, t\in [a,\,b]$ be a given oriented  continuous closed curve,
and $p$ a point of the plane, not on the curve. Then the mapping $\varphi: \Gamma\rightarrow S^1$ given by
$$
\varphi(x(t))=\dfrac{x(t)-p}{|x(t)-p|}, \quad   t\in [a,b],                        \eqno{(1.5)}
$$
is defined to be the position mapping of the  curve $\Gamma$  relative to $p$. When the point on $\Gamma$ goes around the given oriented curve once, its image point $\varphi(x)$ will go around $S^1$ in the same direction with $\Gamma$ a number of times. When moving counter-clockwise or clockwise, we set the sign $+$ or $-$, and we denote it by $deg(\Gamma,\,\,p)$. If $p$ is the origin, we denote it by $deg(\Gamma)$.}

C.H.Deng and S.Q.Zhang $[20]$, X.Su and S.Q.Zhang $[33]$ studies  periodic solutions for a class of planar $N+2$-body problems, they defined the following orbit spaces:
\begin{align}
\Lambda_0=\{ & q\in E_0\,|\, q_i(t+\dfrac{T}{r})=O(\dfrac{2\pi}{r}) q_i(t), \quad i=1,\ldots,N+2;\notag\\                           &q_{i+1}(t)=q_i(t+\dfrac{T}{N}),\,\,i=1,\ldots,N-1,\,\, q_1(t)=q_N(t+\dfrac{T}{N});\notag\\
 & q_i(t+\dfrac{T}{N})=q_i(t), \,\,i=N+1,\,N+2, \forall\, t >0 \tag{1.6}\}
\end{align}
and
\begin{align}
\Lambda=\{&q\in \Lambda_0 \,|\, q_i(t)\neq q_j(t),\, \forall i\neq j,\forall t\in \mathbb{R} ;\notag\\
&deg(q_i(t)-q_j(t))=1,\, 1\leq i\neq j\leq N, deg(q_{N+1}(t)-q_{N+2}(t))=k_1            \tag{1.7}    \},
\end{align}
where
$$
E_0=\{ q=(q_1,q_2,\ldots, q_{N+2}) \,| \,q_i(t)\in W^{1,2}(\mathbb{R}/T\mathbb{Z},\mathbb{R}^2),\,\sum_{i=1}^{N+2} m_i q_i(t)=0  \} ,\eqno{(1.8)}
$$
\[
O(\theta)=\left(
\begin{array}{cc}
\cos{\theta}&-\sin{\theta}\\
\sin{\theta}&\cos{\theta}
\end{array}
\right).
\]
Motivated by their work, we consider $N+3$-body  problems($N>3$,\, $N$ and  $3$ are coprime), the equations of the motion are:
$$
m_i\ddot{q}_i(t)=\dfrac{\partial U(q)}{\partial q_i},\,\,\,i=1,\,\ldots,\,N+3. \eqno{(1.9)}
$$
We define the following orbit spaces :
\begin{align}
\Lambda_1=\{  q\in E_1\,|\,&q_i(t+\frac{T}{r})=O(\frac{2\pi d}{r})q_i(t),\quad i=1,\,\ldots,\,N+3;\notag \\
&q_{i+1}(t)=q_i(t+\frac{T}{N}),\,\,i=1,\, \ldots,\, N,\,\, q_1(t)=q_N(t+\frac{T}{N});  \notag \\
&q_{N+j}(t)=q_{N+j-1}(t+\frac{T}{3}),\,j=2,\,3,\,\, q_{N+1}(t)=q_{N+3}(t+\dfrac{T}{3}) ; \notag \\
&q_i(t+\frac{T}{3})=q_i(t),\,i=1,\ldots,N ; \notag  \\
 & q_j(t+\frac{T}{N})=q_j(t),\, j=N+1, N+2, N+3\}, \tag{1.10}
\end{align}
and
\begin{align}
\Lambda_2=\{q\in \Lambda_1|&q_i(t)\neq q_j(t),  \forall i\neq j, \forall t\in R; \notag\\
&deg(q_i(t)- q_j(t))=k_1, 1\leq i<j\leq N;  \notag\\
&deg (q_{i^\prime}(t)-q_{j^\prime}(t))=k_2, N+1\leq i^\prime<j^\prime\leq N+3\}\notag,
\end{align}
where
$$
E_1=\{ q=(q_1,q_2,\ldots, q_{N+3}) | q_i(t)\in W^{1,2}(\mathbb{R}/T\mathbb{Z},\mathbb{R}^2),\,\sum_{i=1}^{N+3} m_i q_i(t)=0\} .\eqno{(1.11)}
$$

Notice that $r, k_1,k_2,d $ satisfy the following compatible conditions:
$$
k_1=d(mod \,r)\, , k_2=d(mod\, r) , k_1=3s_1 ,k_2=Ns_2, s_1,s_2\in \mathbb{Z}.\eqno{(1.12)}
$$

Since $N$ and $3$ are coprime, we have $(N,3)=1$. In this paper,  we also require $r$  and $3$ coprime,  so $(r,3)=1$.

We get the following theorem:

{\bf Theorem 1.1}\hspace{0.1cm}{\it $(1)$ \,Consider the seven-body problems $(1.9)$ of equal masses, for $r=7,\,k_1=3,\,k_2=-4,\, d=3$, then the global minimizer of $f$ on $\bar{\Lambda}_2$ is a non-collision periodic solution of $(1.9)$.

$(2)$\hspace{0.1cm}Consider the eight-body problems $(1.9)$ of equal masses, for $r=8,\,k_1=3,\,k_2=-5,\,d=3$, then the global minimizer of $f$ on $\bar{\Lambda}_2$ is a non-collision periodic solution of $(1.9)$.

$(3)$\hspace{0.1cm}Consider the ten-body problems $(1.9)$ of equal masses, for $r=10,\,k_1=3,\,k_2=-7,\,d=3$, then the global minimizer of $f$ on $\bar{\Lambda}_2$ is a non-collision periodic solution of $(1.9)$.}

%
%
\section{ Some Lemmas}
\begin{lem}
(Eberlein-Shmulyan$[7]$)\quad  A Banach space $X$ is reflexive if
and only if any bounded sequence in $X$ has a weakly convergent subsequence.
\end{lem}
\begin{lem}
 $([7])$\quad  Let $X$ be a real reflexive Banach space, $M\subset X$ is a weakly closed subset, $f:M\rightarrow R$ is weakly semi-continuous.If $f$ is coercive, that is, $f(x)\rightarrow +\infty $ as $\parallel x\parallel \rightarrow +\infty$, then $f(x)$ attains its infimum on $M$.
\end{lem}

\begin{lem}
$([30])$\quad Let $G$ be a group acting orthogonally on a Hilbert space $H$. Define the fixed point space $F_G=\{x\in H|g\cdot x=x ,\forall g \in G\}$, if $f\in C^1(H,R)$  and satisfies $f(g\cdot x)=f(x)$ for any $g\in G$ and $x\in H$, then the critical point of $f$ restricted on $F_G$ is also a critical point of $f$ on $H$.
\end{lem}
\begin{lem}
$([41])$\quad  Let $q \in W^{1,2}(\mathbb{R} / T\mathbb{Z},\mathbb{R}^n) $
and $ \int_0^Tq(t)\,\mathrm{d}t=0 $, then we have \\
$(i)$.\,Poincare-Wirtinger's inequality:
$$
\int_0^T|\dot{q}(t)|^2\,\mathrm{d}t \geq
{\Big(\frac{2\pi}{T}\Big)}^2 \int_0^T|{q}(t)|^2\,\mathrm{d}t.
\eqno{(2.1)}
$$
$(ii)$.\,  Sobolev's inequality:
$$
\underset{0 \leq t \leq T} {\max} |q(t)|
 = \parallel q \parallel_\infty  \leq  \sqrt {\frac{T}{12}}\big(\int_0^T |\dot
 {q}(t)|^2\mathrm{d}t\big)^{1/2}. \eqno{(2.2)}
$$
\end{lem}
\begin{lem}
 (Gordon[24])(1) Let $x(t)\in W^{1,2}([t_1,t_2],R^k)$ and $x(t_1)=x(t_2)=0$,  Then for any $a>0$, we have
 $$
 \int_{t_1}^{t_2}(\frac{1}{2}|\dot{x}|^2+\frac{a}{|x|})dt\geq
 \frac{3}{2}(2\pi)^{\frac{2}{3}}a^{\frac{3}{2}}(t_2-t_1)^{\frac{1}{3}}.  \eqno{(2.3)}
 $$
 (2)(Long and Zhang$[26]$) Let $x(t) \in W^{1,2}(R/TZ,R^k)$, $\int_0^T x dt=0$, then for any $a>0$,  we have
$$
 \int_0^T(\frac{1}{2}|\dot{x}(t)|^2+\frac{a}{|x|})dt\geq\frac{3}{2}(2\pi)^{\frac{2}{3}}a^{\frac{3}{2}}T^{\frac{1}{3}}.\eqno{(2.4)}
 $$
\end{lem}
\section{Proof of Theorem 1.1}
we consider the system $(1.9)$ of equal masses. Without loss of generality,  we suppose that the masses $m_1=m_2=\cdots=m_{N+3}=1$, and the period $T=1$.

Define $G=\mathbb{Z}_r\times \mathbb{Z}_3\times \mathbb{Z}_N$ and the group action $g=\langle g_1\rangle\times\langle g_2\rangle \times \langle g_3\rangle $  on the space $E_1$:
\begin{align}
g_1(q_1(t),\ldots, q_{N+3}(t)) =(O(-\frac{2\pi d}{r})q_1(t+\frac{1}{r}),\ldots,O(-\frac{2\pi d}{r})q_{N+3}(t+\frac{1}{r})) \tag{3.1}
\end{align}
\begin{align}
&g_2(q_1(t),\ldots, q_{N+3}(t))  \notag \\
&= ( q_1(t+\frac{1}{3}),\ldots , q_N(t+\frac{1}{3}),q_{N+3}(t+\frac{1}{3}), q_{N+1}(t+\frac{1}{3}),q_{N+2}(t+\frac{1}{3})) \tag{3.2}
\end{align}
\begin{align}
&g_3(q_1(t),\ldots, q_{N+3}(t))  \notag \\
&=( q_N(t+\frac{1}{N}),q_1(t+\frac{1}{N})\ldots , q_{N-1}(t+\frac{1}{N}),q_{N+1}(t+\frac{1}{N}), q_{N+2}(t+\frac{1}{N}),q_{N+3}(t+\frac{1}{N})) \tag{3.3}
\end{align}
This implies that $\Lambda_1$ is the fixed point space  of $g$ on $E_1$. Furthermore, for any $g_i$ and $q\in E_1$, we have $ f(g_i\cdot q)=f(q)$ for $i=1,2,3$. Then the Palais symmetry principle implies that the critical point of $f$ restricted on $\Lambda_1$ is also a critical  point of $f$ on $E_1$.

\begin{lem}
The critical point of minimizing the Lagrangian functional $f$  restricted  on $\Lambda_2$ (with winding number restriction) is also a critical point of $f$ on $\Lambda_1$, then it is also the solution of $(1.9)$.
\end{lem}

The proof is similar to that of Lemma $3.1$ in $[21]$, we omit it.

By $q_i(t)=O(-\dfrac{2\pi d}{r})q_i(t+\dfrac{1}{r})(i=1,\cdots,N+3)\,$,  we have

$$
\int_0^1q_i(t) dt=0.
$$
Then the Lemma $2.4$
$$
\int_0^1|\dot{q}_i(t)|^2 dt\geq (2\pi)^2\int_0^1|q_i(t)|^2dt.
$$

Hence $f(q)$ is coercive on $\bar{\Lambda}_2$. It is easy to see that $\bar{\Lambda}_2$ is a weakly closed subset.Fatou's lemma implies that $f(q)$ is a weakly lower semi-continuous. Then by Lemma $2.2$, $f(q)$  attains $\inf{\{f(q)| q\in \bar{\Lambda}_2\}}$.  Similar to Lemma $3.2$ in $[21]$, we can obtain the following lemma.

\begin{lem}
The limit curve $q(t)=(q_1(t),q_2(t),\ldots, q_{N+3}(t))\in \partial{\Lambda_2}$ of a sequence $ q^l(t)=(q^l_1(t),q^l_2(t),\ldots, q^l_{N+3}(t))\in \Lambda_2$  may either have collisions between some two point masses or has the same winding number $(i.e. deg(q_i(t)-q_j(t))=k_1, 1\leq i\neq j\leq N; deg(q_{i^\prime}(t)-q_{j^\prime}(t))=k_2,N+1\leq i^\prime\neq j^\prime\leq N+3).$
\end{lem}

In the following, we prove that the minimizer of $f$ is a non-collision solutions of  the system $(1.9).$

Since $ \sum_{i=1}^{N+3}q_i=0$, by the Lagrangian identity, we have

$$
f(q)=\frac{1}{N+3}\sum_{1\leq i<j\leq N+3}\int_0^1(\frac{1}{2}\,|\dot{q}_i-\dot{q}_j|^2+\frac{N+3}{|q_i-q_j|} )dt \eqno{(3.4)}
$$

Notice that each term on the right hand side of $(3.4)$ is a Lagrangian action for a suitable two body problem, which is a key step for the lower bound estimate on the collision set.

We estimate the infimum of the action functional on the collision set. Since  the symmetry for a two-body problem implies that the Lagrangian action on a collision solution is greater than that on the  non-collision solution, and the more collisions there are, the greater the Lagrangian is. We  only assume that the two bodies collide at some moment $t_0$, without loss of generality, let $t_0=0$, we will sufficiently use the symmetries of collision orbits.

since  $q\in \bar{\Lambda}_2$,  we have
\begin{align}
q_i(t+\dfrac{1}{r})=O(\dfrac{2\pi d}{r})q_i(t), \,i=1,\ldots, N+3; \tag{3.5}
\end {align}
\begin{align}
&q_{i+1}(t)=q_i(t+\dfrac{1}{N}),i=1,\ldots,N-1,\,\,q_1(t)=q_N(t+\dfrac{1}{N});\tag{3.6}\\
&q_{N+2}(t)=q_{N+1}(t+\dfrac{1}{3}),\,q_{N+3}(t)=q_{N+2}(t+\dfrac{1}{3}),\,q_{N+1}(t)=q_{N+3}(t+\dfrac{1}{3});\tag{3.7}
\end{align}
\begin{align}
&q_i(t+\dfrac{1}{3})=q_i(t),\,i=1,\ldots,N ;\tag{3.8}\\
&q_j(t+\dfrac{1}{N})=q_j(t),\,j=N+1,N+2,N+3.\tag{3.9}
\end{align}

Case $1$:  $q_1,q_2$ collide at $t=0$.

By $(3.5)$, we can deduce $q_1,q_2$ collide at $t=\dfrac{i}{r} ,\, i=0,\ldots, r-1.$

Furthermore, by $(3.8)$,  we can deduce $q_1,q_2$ collide at
$$
t= \dfrac{i}{r},\,\,\dfrac{i}{r}+\dfrac{1}{3},\,\, \dfrac{i}{r}+ \dfrac{2}{3}\,(mod\, 1). \eqno{(3.10)}
$$
From $(3.6)$ and $(3.10)$, we have
\begin{align*}
&q_2,q_3 \,\text{collide at}\, \dfrac{i}{r}+\dfrac{N-1}{N},\,\,\dfrac{i}{r}+\dfrac{1}{3}+\dfrac{N-1}{N},\,\,\dfrac{i}{r}+\dfrac{2}{3}
+\dfrac{N-1}{N}(mod\,1),\,\,i=0,\ldots, r-1,\\
&q_3,q_4 \,\text{collide at}\,
\dfrac{i}{r}+\dfrac{N-2}{N},\,\,\dfrac{i}{r}+\dfrac{1}{3}+\dfrac{N-2}{N},\,\,\dfrac{i}{r}+\dfrac{2}{3}+\dfrac{N-2}{N}(mod\,1),
\,\,i=0,\ldots,r-1,
\end{align*}
\begin{align*}
&\vdots\\
&q_{N-1},q_{N} \,\text{collide at}\, \dfrac{i}{r}+\dfrac{2}{N},\,\,\dfrac{i}{r}+\dfrac{1}{3}+\dfrac{2}{N},\,\,\dfrac{i}{r}+\dfrac{2}{3}
+\dfrac{2}{N}(mod\,1),\,i=0,\ldots, r-1,\\
&q_N,q_1\, \text{collide at}\, \dfrac{i}{r}+\dfrac{1}{N},\,\,\dfrac{i}{r}+\dfrac{1}{3}+\dfrac{1}{N},\,\,\dfrac{i}{r}+\dfrac{2}{3}
+\dfrac{1}{N}(mod\,1),\,i=0,\ldots, r-1.
\end{align*}
\begin{lem}
$\forall\,\, 0\leq i,\,j\leq r-1,\,0\leq k\leq 2, \,\,(i-j)^2+k^2\neq 0$, we have
$$
\dfrac{i}{r}\neq \dfrac{j}{r}+\dfrac{k}{3}(mod\,1)  \eqno{(3.11)}
$$
\end{lem}

\begin{proof}

 If there exist $0\leq i_0,j_0\leq r-1 ,0\leq k_0\leq 2, (i_0-j_0)^2+k_0^2\neq 0$ such that
 $$
  \dfrac{i_0}{r}= \dfrac{j_0}{r}+\dfrac{k_0}{3}(mod\,1).
 $$

Then  we have
$$
 1| (\dfrac{j_0}{r}+\dfrac{k_0}{3}-\dfrac{i_0}{r}).
$$

Since
$$
\dfrac{j_0}{r}+\dfrac{k_0}{3}-\dfrac{i_0}{r}\geq -\dfrac{r-1}{r}=-1+\dfrac{1}{r}>-1,
$$

and
$$
\dfrac{j_0}{r}+\dfrac{k_0}{3}-\dfrac{i_0}{r}\leq  \dfrac{r-1}{r}+\dfrac{2}{3}<2,
$$

we can deduce
\begin{center}
$\dfrac{j_0}{r}+\dfrac{k_0}{3}-\dfrac{i_0}{r}=0$\quad or \quad $\dfrac{j_0}{r}+\dfrac{k_0}{3}-\dfrac{i_0}{r}=1.$
\end{center}

If $\dfrac{j_0}{r}+\dfrac{k_0}{3}-\dfrac{i_0}{r}=0$,  then $ 3(i_0-j_0)=k_0r$. When $k_0=0$, we get $i_0=j_0$, which is a contradiction with our assumptions on the $i_0,\,j_0,\,k_0$; when $k_0\neq 0$, notice $0<k_0\leq 2$, we can deduce $3|r$, which is a contradiction since $(r,3)=1.$

If $\dfrac{j_0}{r}+\dfrac{k_0}{3}-\dfrac{i_0}{r}=1$,  then $3(j_0-i_0)=(3-k_0)r$. When $k_0=0$, we get $r=j_0-i_0$, which is a contradiction since $-r+1\leq j_0-i_0\leq r-1$; when $k_0\neq 0$, notice $1\leq 3-k_0<3$, we can deduce $3|r$, which is also a contradiction since  $(r,3)=1.$
\end{proof}
By $(3.10)$ and Lemma $3.3$,  we know that $q_1,q_2$ collide at
$$
t_i=\dfrac{i}{3r},\,\,i=0,\ldots,\,3r-1.  \eqno{(3.12)}
$$
Then by Lemma $2.5$, $(3.12)$, we have
\begin{align}
&\int_0^1(\dfrac{1}{2}|\dot{q}_1(t)-\dot{q}_2(t)|^2+\dfrac{N+3}{|q_1(t)-q_2(t)|})dt\notag\\
&=\sum_{i=0}^{3r-1}\int_{t_i}^{t_{i+1}}(\dfrac{1}{2}|\dot{q}_1(t)-\dot{q}_2(t)|^2+\dfrac{N+3}{|q_1(t)-q_2(t)|})dt\notag\\
&\geq  \dfrac{3}{2}\times (2\pi)^{\frac{2}{3}}(N+3)^{\frac{2}{3}}3r(\dfrac{1}{3r})^{\frac{1}{3}}. \tag{3.13}
\end{align}

From $(3.6)$ and $(3.12)$, we have
\begin{align}
&q_2,q_3 \,\text{collide at}\, \dfrac{i}{3r}+\dfrac{N-1}{N}(mod\,1),\,\,i=0,\ldots, 3r-1,\notag\\
&q_3,q_4 \,\text{collide at}\, \dfrac{i}{3r}+\dfrac{N-2}{N}(mod\,1),\,\,i=0,\ldots, 3r-1,\notag\\
&\vdots\notag\\
&q_{N-1},q_N\,\text{collide at}\, \dfrac{i}{3r}+\dfrac{2}{N}(mod\,1),\,\,i=0,\ldots, 3r-1,\tag{3.14}
\end{align}
\begin{align}
q_{N},q_1 \,\text{collide at}\, \dfrac{i}{3r}+\dfrac{1}{N}(mod\,1),\,\,i=0,\ldots, 3r-1.\tag{3.15}
\end{align}
\begin{lem}
$\forall\, 0\leq i,i^\prime\leq 3r-1, 1\leq j,j^\prime\leq N-1, (i-i^\prime)^2+(j-j^\prime)^2\neq 0$, we have
$$
\dfrac{i}{3r}+\dfrac{j}{N}\neq \dfrac{i^\prime}{3r}+\dfrac{j^\prime}{N}(mod\,1). \eqno{(3.16)}
$$
\end{lem}
 The proof is similar to Lemma $3.3$.

 {\bf Remark 3.1}\quad From Lemma $3.4$,\,\,$\forall\,0\leq i,i^\prime\leq r-1,\,\, 1\leq j,j^\prime\leq N-1,\,\, 0\leq k,k^\prime\leq 2, (i-i^\prime)^2
+(j-j^\prime)^2+(k-k^\prime)^2\neq 0$, we have
\begin{align*}
\dfrac{i}{r}+\dfrac{j}{N}+\dfrac{k}{3}\neq \dfrac{i^\prime}{r}+\dfrac{j^\prime}{N}+\dfrac{k^\prime}{3}(mod\,1).
\end{align*}
By Lemma $2.5$, Lemma $3.4$ and $(3.15)$, we have
\begin{align}
&\int_0^1(\dfrac{1}{2}|\dot{q}_{j+1}(t)-\dot{q}_{j+2}(t)|^2+\dfrac{N+3}{|q_{j+1}(t)-q_{j+2}(t)|})dt\notag\\
&\geq  \dfrac{3}{2}\times (2\pi)^{\frac{2}{3}}(N+3)^{\frac{2}{3}}3r(\dfrac{1}{3r})^{\frac{1}{3}},\,\,\,(j=1,\ldots,N-2), \tag{3.17}
\end{align}
\begin{align}
&\int_0^1(\dfrac{1}{2}|\dot{q}_N(t)-\dot{q}_1(t)|^2+\dfrac{N+3}{|q_N(t)-q_1(t)|})dt\notag\\
&\geq  \dfrac{3}{2}\times (2\pi)^{\frac{2}{3}}(N+3)^{\frac{2}{3}}3r(\dfrac{1}{3r})^{\frac{1}{3}}. \tag{3.18}
\end{align}
Let
\begin{align*}M_1=\sum_{j=0}^{N-2}&\int_0^1(\dfrac{1}{2}|\dot{q}_{j+1}(t)-\dot{q}_{j+2}(t)|^2
+\dfrac{N+3}{|q_{j+1}(t)-q_{j+2}(t)|})dt+\\
&\int_0^1(\dfrac{1}{2}|\dot{q}_N(t)-\dot{q}_1(t)|^2+\dfrac{N+3}{|q_N(t)-q_1(t)|})dt.
\end{align*}

Then by $(3.13),(3.17),(3.18)$, Lemma $2.5$, and notice that  $\forall\, 1\leq i\leq N,\,N+1\leq j\leq N+3,\,\, \int_0^{\frac{1}{3}}q_i(t)dt=0,\,\, \int_0^{\frac{1}{N}}q_j(t)dt=0$,  so we have
\begin{align}
f(q)&=\frac{1}{N+3}\sum_{1\leq i<j\leq N+3}\int_0^1(\frac{1}{2}\,|\dot{q}_i(t)-\dot{q}_j(t)|^2+\frac{N+3}{|q_i(t)-q_j(t)|})dt \notag\\
&=\dfrac{1}{N+3}\{\,\, M_1+
[\sum_{1\leq i<j\leq N}\int_0^1(\frac{1}{2}\,|\dot{q}_i(t)-\dot{q}_j(t)|^2+\frac{N+3}{|q_i(t)-q_j(t)|})dt-M_1\,]+ \notag\\
&\qquad \qquad \sum_{1\leq i\leq N, 1\leq j\leq 3 }\int_0^1(\frac{1}{2}\,|\dot{q}_i(t)-\dot{q}_{N+j}(t)|^2+\frac{N+3}{|q_i(t)-q_{N+j}(t)|})dt+\notag\\
&\qquad \qquad \sum_{N+1\leq i<j\leq N+3 }\int_0^1(\frac{1}{2}\,|\dot{q}_i(t)-\dot{q}_{j}(t)|^2+\frac{N+3}{|q_i(t)-q_{j}(t)|})dt\,\,\}\notag\\
&\geq\dfrac{3}{2}\times(\frac{4\pi^2}{N+3})^{\frac{1}{3}}[\, N\times3r
(\frac{1}{3r})^{\frac{1}{3}}+3\times (\frac{1}{3})^{\frac{1}{3}}(C_N^2-N) +3N+3 N(\frac{1}{N})^{\frac{1}{3}}\, ]\notag\\
&\stackrel{\triangle}{=} A.\tag{3.19}
\end{align}
In the following cases, we firstly study the cases under $N$ is even.

Case $2$: $q_1,q_{k+2}(k=1,\ldots, \dfrac{N}{2}-2)$ collide at $t=0$.

By $(3.5)$, we can deduce $q_1,q_{k+2}(k=1,\ldots, \dfrac{N}{2}-2)$ collide at $t=\dfrac{i}{r} ,\, i=0,\ldots, r-1.$

Then  by $(3.8)$ , $q_1,q_{k+2}$ collide at

$$t= \dfrac{i}{r},\,\,\dfrac{i}{r}+\dfrac{1}{3},\,\, \dfrac{i}{r}+ \dfrac{2}{3}\,(mod\, 1), \,\,i=0,\cdots, r-1. \eqno{(3.20)}
$$

From Lemma $3.3$,  we get  $q_1,q_{k+2}$ collide at

$$
t=\dfrac{i}{3r}, i=0,\ldots,3r-1. \eqno{(3.21)}
$$
Then by $(3.8)$,  we have
\begin{align}
&q_2,q_{k+3} \,\text{collide at}\, t=\dfrac{i}{3r}+\dfrac{N-1}{N}(mod\,1),\, i=0,\ldots,3r-1,\notag\\
&q_3,q_{k+4} \,\text{collide at}\, t=\dfrac{i}{3r}+\dfrac{N-2}{N}(mod\,1),\, i=0,\ldots,3r-1,\notag\\
&\vdots\notag\\
&q_{N-k-1},q_N\,\text{collide at}\, t=\dfrac{i}{3r}+\dfrac{k+2}{N}(mod\,1),\, i=0,\ldots,3r-1,\notag\\
&q_{N-k},q_1,\text{collide at}\, t=\dfrac{i}{3r}+\dfrac{k+1}{N}(mod\,1),\, i=0,\ldots,3r-1,\notag\\
&q_{N-k+1},q_2\text{collide at}\, t=\dfrac{i}{3r}+\dfrac{k}{N}(mod\,1),\, i=0,\ldots,3r-1,\notag
\end{align}
\begin{align}
&\vdots\notag\\
&q_N,q_{k+1}\text{collide at}\, t=\dfrac{i}{3r}+\dfrac{1}{N}(mod\,1),\, i=0,\ldots,3r-1.\tag{3.22}
\end{align}
Then by Lemma $2.5$, Lemma $3.3$, Lemma $3.4$, $(3.21)-(3.22)$, we have
\begin{align}
f(q)&\geq\dfrac{3}{2}\times(\frac{4\pi^2}{N+3})^{\frac{1}{3}}[\, N\times3r
(\frac{1}{3r})^{\frac{1}{3}}+3\times (\frac{1}{3})^{\frac{1}{3}}(C_N^2-N) +3N+3 N(\frac{1}{N})^{\frac{1}{3}}\, ]\notag\\
&=A.\tag{3.23}
\end{align}
Case 3: $q_1,q_{\frac{N}{2}+1}$ collide at $t=0$.

By $(3.5), (3.6),\,\,(3.8)$,  $q_1,q_{\frac{N}{2}+1}$  collide at
\begin{align}
t=& \dfrac{i}{r},\,\dfrac{i}{r}+\dfrac{1}{3},\,\,\dfrac{i}{r}+\dfrac{2}{3},\,\notag\\
&\dfrac{i}{r}+\dfrac{\frac{N}{2}}{N},\,\frac{i}{r}+\dfrac{1}{3}+\dfrac{\frac{N}{2}}{N},\,\,\dfrac{i}{r}+\dfrac{2}{3}
+\dfrac{\frac{N}{2}}{N}(mod\,1), i=0,\ldots, r-1. \tag{3.24}
\end{align}
 Simplify  $(3.24)$ , we get  $q_1,q_{\frac{N}{2}+1}$ collide at

\begin{align}
t=\dfrac{i}{r}+\dfrac{j}{6},\,i=0,\ldots, r-1,\,j=0,\ldots,5 \tag{3.25}
\end{align}

\begin{lem}
$\forall\, 0\leq i,i^\prime\leq r-1, 0\leq j,j^\prime\leq 5 , (i-i^\prime)^2+(j-j^\prime)^2\neq 0$, we have
$$
\dfrac{i}{r}+\dfrac{j}{6}\neq\dfrac{i^\prime}{r}+\dfrac{j^\prime}{6}(mod\,1) \eqno{(3.26)}
$$
\end{lem}
\begin{proof}

 If there exist $0\leq i_0,i_1\leq r-1,\,0\leq j_0,j_1\leq 5,(i_0-i_1)^2+(j_0-j_1)^2\neq 0$  such that
\begin{align}
\dfrac{i_0}{r}+\dfrac{j_0}{6}= \dfrac{i_1}{r}+\dfrac{j_1}{6} (mod1)  \tag{3.27}
\end{align}

Since
\begin{align}
\dfrac{i_1}{r}+\dfrac{j_1}{6}-\dfrac{i_0}{r}-\dfrac{j_0}{6}\geq -\dfrac{r-1}{r}-\dfrac{5}{6}>-2, \notag \\
\dfrac{i_1}{r}+\dfrac{j_1}{6}-\dfrac{i_0}{r}-\dfrac{j_0}{6}\leq \dfrac{r-1}{r}+\dfrac{5}{6}<2 , \notag
\end{align}
 then we deduce
 \begin{center}
 $ \dfrac{i_1}{r}+\dfrac{j_1}{6}- \dfrac{i_0}{r}-\dfrac{j_0}{6}=-1$ ,\,\,or $ \dfrac{i_1}{r}+\dfrac{j_1}{6}- \dfrac{i_0}{r}-\dfrac{j_0}{6}=0$,\,\,or $ \dfrac{i_1}{r}+\dfrac{j_1}{6}- \dfrac{i_0}{r}-\dfrac{j_0}{6}=1$.
\end{center}

If $\dfrac{i_1}{r}+\dfrac{j_1}{6}- \dfrac{i_0}{r}-\dfrac{j_0}{6}=-1$, we have $r(6+j_1-j_0)=6(i_0-i_1) $. When $i_0=i_1$, which is a contradiction since $r(6+j_1-j_0)\neq 0$ ; when $i_0\neq i_1$ and $j_0=j_1$ , we can deduce $r=i_0-i_1$, which is a contradiction since $-r+1\leq i_0-i_1\leq r-1$; when $i_0\neq i_1$ and $j_0\neq j_1$, we can deduce $ 6|r$, which is a contradiction since $(r,3)=1$.

We can use similar arguments to prove $\dfrac{i_1}{r}+\dfrac{j_1}{6}- \dfrac{i_0}{r}-\dfrac{j_0}{6}\neq 0$ and $\dfrac{i_1}{r}+\dfrac{j_1}{6}- \dfrac{i_0}{r}-\dfrac{j_0}{6}\neq 1$.
\end{proof}
From $(3.25)$ and $(3.26)$, we can deduce $q_1,q_{\frac{N}{2}+1}$ collide at
$$
t_i=\dfrac{i}{6r},\,\,r=0, \ldots, 6r-1.  \eqno{(3.28)}
$$

Then by Lemma $2.5$ and $(3.28)$, we have
\begin{align}
&\int_0^1(\dfrac{1}{2}|\dot{q}_1(t)-\dot{q}_{\frac{N}{2}+1}(t)|^2+\dfrac{N+3}{|q_1(t)-q_{\frac{N}{2}+1}(t)|})dt\notag\\
&=\sum_{i=0}^{6r-1}\int_{t_i}^{t_{i+1}}(\dfrac{1}{2}|\dot{q}_1(t)-\dot{q}_{\frac{N}{2}+1}(t)|^2+
\dfrac{N+3}{|q_1(t)-q_{\frac{N}{2}+1}(t)|})dt\notag\\
&\geq  \dfrac{3}{2}\times (2\pi)^{\frac{2}{3}}(N+3)^{\frac{2}{3}}6r(\dfrac{1}{6r})^{\frac{1}{3}}. \tag{3.29}
\end{align}
By $(3.6)$, $(3.28)$, we have
\begin{align}
&q_2 ,q_{\frac{N}{2}+2},\,\text{collide at}\, t=\dfrac{i}{6r}+\dfrac{\frac{N}{2}-1}{N}, \, i=0,\ldots,6r-1,\notag\\
&q_3, q_{\frac{N}{2}+3},\,\text{collide at}\, t=\dfrac{i}{6r}+\dfrac{\frac{N}{2}-2}{N}, \, i=0,\ldots,6r-1,\notag\\
&\vdots\notag\\
&q_{\frac{N}{2}},q_N\,\text{collide at}\, t=\dfrac{i}{6r}+\dfrac{1}{N}, \, i=0,\ldots,6r-1.\tag{3.30}
\end{align}
\begin{lem}
$\forall\,0\leq i,i^\prime\leq 6r-1, 1\leq j,j^\prime\leq \dfrac{N}{2}-1 ,\,(i-i^\prime)^2+(j-j^\prime)^2\neq0$, we have
\end{lem}

$$
\dfrac{i}{6r}+\dfrac{j}{N}\neq \dfrac{i^\prime}{6r}+\dfrac{j^\prime}{N}.    \eqno{(3.31)}
$$
The proof is similar to Lemma $3.5$.

By Lemma $2.5$, Lemma $3.6$, $(3.30)-(3.31)$, we have
\begin{align}
&\int_0^1(\dfrac{1}{2}|\dot{q}_{j+1}(t)-\dot{q}_{\frac{N}{2}+j+1}(t)|^2+\dfrac{N+3}{|q_{j+1}(t)-q_{\frac{N}{2}+j+1}(t)|})dt\notag\\
&\geq  \dfrac{3}{2}\times (2\pi)^{\frac{2}{3}}(N+3)^{\frac{2}{3}}6r(\dfrac{1}{6r})^{\frac{1}{3}}  \quad(j=1,\ldots,\frac{N}{2}-1).  \tag{3.32}
\end{align}
Let
\begin{align*}
M_2=\sum_{j=0}^{\frac{N}{2}-1}\int_0^1 (\frac{1}{2}\,|\dot{q}_{j+1}(t)-\dot{q}_{\frac{N}{2}+j+1}(t)|^2+\frac{N+3}{|q_{j+1}(t)-q_{\frac{N}{2}+j+1}(t)|})dt
\end{align*}
Then from Lemma $2.5$, Lemma $3.6$,  $(3.29)$ and $(3.32)$, we obtain
\begin{align}
f(q)&=\frac{1}{N+3}\sum_{1\leq i<j\leq N+3}\int_0^1(\frac{1}{2}\,|\dot{q}_i(t)-\dot{q}_j(t)|^2+\frac{N+3}{|q_i(t)-q_j(t)|})dt \notag\\
&=\dfrac{1}{N+3}\{\,\,M_2+ [\sum_{1\leq i<j\leq N} \int_0^1(\frac{1}{2}\,|\dot{q}_i(t)-\dot{q}_j(t)|^2+\frac{N+3}{|q_i(t)-q_j(t)|} )dt-M_2]+ \notag\\
&\quad \qquad \sum_{1\leq i\leq N,1\leq j\leq 3} \int_0^1(\frac{1}{2}\,|\dot{q}_i(t)-\dot{q}_{N+j}(t)|^2+\frac{N+3}{|q_i(t)-q_{N+j}(t)|} )dt + \notag\\
&\quad \qquad \sum_{N+1\leq i<j\leq N+3}\int_0^1(\frac{1}{2}\,|\dot{q}_i(t)-\dot{q}_j(t)|^2+\frac{N+3}{|q_i(t)-q_j(t)|} )dt \,\,\} \notag\\
&\geq\dfrac{3}{2}\times(\frac{4\pi^2}{N+3})^{\frac{1}{3}}[ \,\dfrac{N}{2}\times 6r
(\frac{1}{6r})^{\frac{1}{3}}+3\times (\frac{1}{3})^{\frac{1}{3}}(C_N^2-\dfrac{N}{2}) +3N+3 N(\frac{1}{N})^{\frac{1}{3}} \,]\notag\\
&\stackrel{\triangle}{=} B.\tag{3.33}
\end{align}
Finally, we study the cases under $N$ is odd.

Case $2^{\prime}$: $q_1,q_{k+2}(k=1,\ldots,{\frac{N+1}{2}}-2)$ collide at  $t=0$.

By  $(3.5),(3.8)$,
$q_1,q_{k+2}(k=1,\ldots,\frac{N+1}{2}-2)$  collide at
$$
t=\dfrac{i}{r},\,\, \dfrac{i}{r}+\dfrac{1}{3},\,\,\dfrac{i}{r}+\dfrac{2}{3}(mod\,1),\,i=0,\ldots,r-1, \eqno{(3.34)}
$$

from Lemma $3.3$, we get $q_1,q_{k+2}(k=1,\ldots,\frac{N+1}{2}-2)$  collide at
$$
t=\dfrac{i}{3r},\,\,i=0,\ldots,3r-1, \eqno{(3.35)}
$$
then by $(3.6)$, we have
\begin{align}
&q_2,q_{k+3} \,\text{collide at}\, t=\dfrac{i}{3r}+\dfrac{N-1}{N}(mod\,1),\, i=0,\ldots,3r-1,\notag\\
&q_3,q_{k+4} \,\text{collide at}\, t=\dfrac{i}{3r}+\dfrac{N-2}{N}(mod\,1),\, i=0,\ldots,3r-1,\notag\\
&\vdots\notag\\
&q_{N-k-1},q_N\,\text{collide at}\, t=\dfrac{i}{3r}+\dfrac{k+2}{N}(mod\,1),\, i=0,\ldots,3r-1,\notag\\
&q_{N-k},q_1,\text{collide at}\, t=\dfrac{i}{3r}+\dfrac{k+1}{N}(mod\,1),\, i=0,\ldots,3r-1,\notag\\
&q_{N-k+1},q_2\text{collide at}\, t=\dfrac{i}{3r}+\dfrac{k}{N}(mod\,1),\, i=0,\ldots,3r-1,\notag\\
&\vdots\notag\\
&q_N,q_{k+1}\text{collide at}\, t=\dfrac{i}{3r}+\dfrac{1}{N}(mod\,1),\, i=0,\ldots,3r-1.\tag{3.36}
\end{align}
Then by Lemma $2.5$, Lemma $3.4$, $(3.35), (3.36)$, we have
\begin{align}
f(q)\geq &\dfrac{3}{2}\times (\dfrac{4\pi^2}{N+3})^{\frac{1}{3}}[\, N\times 3r (\dfrac{1}{3r})^{\frac{1}{3}}+ 3\times(\dfrac{1}{3})^{\frac{1}{3}}(C_N^2-N)+3N+3N(\dfrac{1}{N})^{\frac{1}{3}}\, ] \notag\\
&=A.\tag{3.37}
\end{align}

Case $4$:  $q_{N+1}, q_1$ collide at $t=0$.

By $(3.5)$,  we have

$q_{N+1}, q_1$  collide at
$$
t=\dfrac{i}{r}, \quad i=0,\ldots,r-1.       \eqno{(3.38)}
$$
Then by Lemma $2.5$, $(3.37)$,  we have
\begin{align}
\int_0^1&(\dfrac{1}{2}|\dot{q}_1(t)-\dot{q}_{N+1}(t)|^2+\dfrac{N+3}{|q_1(t)-q_{N+1}(t)|})dt\notag\\
&=\sum_{i=0}^{r-1}\int_{t_i}^{t_{i+1}}(\dfrac{1}{2}|\dot{q}_1(t)-\dot{q}_{N+1}(t)|^2+
\dfrac{N+3}{|q_1(t)-q_{N+1}(t)|})dt\notag\\
&\geq \dfrac{3}{2}\times (4\pi^2)(N+3)^{\frac{2}{3}}r(\dfrac{1}{r})^{\frac{1}{3}}.  \tag{3.39}
\end{align}

From $(3.38), (3.5)-(3.9)$, we can  obtain

$q_{N+2},q_1,$  collide at $t=\dfrac{i}{r}+\dfrac{2}{3}(mod\,1)$, \quad $q_{N+3},q_1$ collide at $t=\dfrac{i}{r}+\dfrac{1}{3}(mod\,1)        ,\,\,i=0,\ldots,r-1,$

$q_{N+1}, q_{2}$ collide at $\dfrac{i}{r}+\dfrac{N-1}{N}(mod\,1)$,\quad $q_{N+2},q_{2}$ collide at $ \dfrac{i}{r}+\dfrac{N-1}{N}+\dfrac{2}{3}(mod\,1)$, \quad $q_{N+3},q_{2}$ collide at $ \dfrac{i}{r}+\dfrac{N-1}{N}+\dfrac{1}{3}(mod\,1),\,i=0,\ldots,r-1,$

\vdots

$q_{N+1}, q_{N-1}$ collide at $\dfrac{i}{r}+\dfrac{2}{N}(mod\,1)$,\quad $q_{N+2},q_{N-1} $ collide at $ \dfrac{i}{r}+\dfrac{2}{N}+\dfrac{2}{3}(mod\,1)$, \quad $q_{N+3},q_{N-1}$ collide at $ \dfrac{i}{r}+\dfrac{2}{N}+\dfrac{1}{3}(mod\,1),\,i=0,\ldots,r-1,$

$q_{N+1}, q_{N}$ collide at $\dfrac{i}{r}+\dfrac{1}{N}(mod\,1)$,\quad $q_{N+2},q_{N} $ collide at $ \dfrac{i}{r}+\dfrac{1}{N}+\dfrac{2}{3}(mod\,1)$, \quad $q_{N+3},q_{N}$ collide at $ \dfrac{i}{r}+\dfrac{1}{N}+\dfrac{1}{3}(mod\,1) ,\,i=0,\ldots,r-1.$

Then by Lemma $2.5$,  Lemma $3.3$, Remark $3.1$, we have  $\forall\, 0\leq i\leq r-1, 1\leq j\leq 3,$
\begin{align}
\int_0^1&(\dfrac{1}{2}|\dot{q}_i(t)-\dot{q}_{N+j}(t)|^2+\dfrac{N+3}{|q_i(t)-q_{N+j}(t)|})dt\notag\\
&\geq \dfrac{3}{2}\times (4\pi^2)(N+3)^{\frac{2}{3}}r(\dfrac{1}{r})^{\frac{1}{3}}.  \tag{3.40}
\end{align}
So we  get
\begin{align}
f(q)&=\frac{1}{N+3}\sum_{1\leq i<j\leq N+3}\int_0^1(\frac{1}{2}\,|\dot{q}_i(t)-\dot{q}_j(t)|^2+\frac{N+3}{|q_i(t)-q_j(t)|})dt \notag
\end{align}
\begin{align}
&=\dfrac{1}{N+3}(\,\,\sum_{\stackrel{1\leq i\leq N}{1\leq j\leq 3}}\int_0^1(\dfrac{1}{2}|\dot{q}_i(t)-\dot{q}_{N+j}(t)|^2
+\dfrac{N+3}{|q_i(t)-q_{N+j}(t)|})dt+\notag\\
&\qquad \qquad \sum_{1\leq i<j\leq N} \int_0^1(\frac{1}{2}\,|\dot{q}_i(t)-\dot{q}_j(t)|^2+\frac{N+3}{|q_i(t)-q_j(t)|} )dt  +\notag\\
&\qquad \qquad \sum_{N+1\leq i<j\leq N+3}\int_0^1(\frac{1}{2}\,|\dot{q}_i(t)-\dot{q}_j(t)|^2+\frac{N+3}{|q_i(t)-q_j(t)|} )dt  \,\,)\notag\\
&\geq\dfrac{3}{2}\times(\frac{4\pi^2}{N+3})^{\frac{1}{3}} [\,3N\times r(\dfrac{1}{r})^{\frac{1}{3}}+3\times (\dfrac{1}{3})^{\frac{1}{3}}C_N^2 +3N(\dfrac{1}{N})^{\frac{1}{3}}\,]\notag\\
&\stackrel{\triangle}{=}C.\tag{3.41}
\end{align}
Case $5$:  $q_{N+1}, q_{N+2}$ collide at $t=0$.

Then by $(3.5),(3.9)$, we deduce

$q_{N+1},q_{N+2}$ collide at
\begin{align}
t=\dfrac{i}{r}+\dfrac{j}{N}(mod\,1),\,i=0,\ldots r-1,\,j=0,\ldots,N-1.  \tag{3.42}
\end{align}

From Remark $3.1$,  and $(3.42)$, we can deduce $q_{N+1}, q_{N+2}$ collide at
$$
t_i=\dfrac{i}{Nr},\quad  i=0,\ldots,Nr-1.  \eqno{(3.43)}
$$
Then we have
\begin{align}
&\int_0^1(\dfrac{1}{2}|\dot{q}_{N+1}(t)-\dot{q}_{N+2}(t)|^2+\dfrac{N+3}{|q_{N+1}(t)-q_{N+2}(t)|})dt\notag\\
&=\sum_{i=0}^{Nr-1}\int_{t_i}^{t_{i+1}}(\dfrac{1}{2}|\dot{q}_{N+1}(t)-\dot{q}_{N+2}(t)|^2+
\dfrac{N+3}{|q_{N+1}(t)-q_{N+2}(t)|})dt\notag\\
&\geq \dfrac{3}{2}\times (4\pi^2)(N+3)^{\frac{2}{3}}Nr(\dfrac{1}{Nr})^{\frac{1}{3}}.  \tag{3.44}
\end{align}
By$(3.7)$, we deduce $q_{N+2},q_{N+3}$, collide at

\begin{align}
t=\dfrac{i}{Nr}+\dfrac{2}{3},\quad i=0,\ldots,Nr-1, \notag \tag{3.45}
\end{align}

$q_{N+3},q_{N+1}$ collide at
\begin{align}
t=\dfrac{i}{Nr}+\dfrac{1}{3} , \quad i=0,\ldots,Nr-1. \notag \tag{3.46}
\end{align}
Then by Lemma $2.5$, Remark $3.1$, $(3.45)$, and $(3.46)$, we have
\begin{align}
&\int_0^1(\dfrac{1}{2}|\dot{q}_{N+2}(t)-\dot{q}_{N+3}(t)|^2+\dfrac{N+3}{|q_{N+2}(t)-q_{N+3}(t)|})dt\notag\\
&\geq \dfrac{3}{2}\times (4\pi^2)(N+3)^{\frac{2}{3}}Nr(\dfrac{1}{Nr})^{\frac{1}{3}}  \tag{3.47}
\end{align}
\begin{align}
&\int_0^1(\dfrac{1}{2}|\dot{q}_{N+3}(t)-\dot{q}_{N+1}(t)|^2+\dfrac{N+3}{|q_{N+3}(t)-q_{N+1}(t)|})dt\notag\\
&\geq \dfrac{3}{2}\times (4\pi^2)(N+3)^{\frac{2}{3}}Nr(\dfrac{1}{Nr})^{\frac{1}{3}}.  \tag{3.48}
\end{align}
So, we obtain
\begin{align}
f(q)&=\frac{1}{N+3}\sum_{1\leq i<j\leq N+3}\int_0^1(\frac{1}{2}\,|\dot{q}_i(t)-\dot{q}_j(t)|^2+\frac{N+3}{|q_i(t)-q_j(t)|})dt \notag\\
&=\dfrac{1}{N+3}( \sum_{N+1\leq i<j\leq N+3}\int_0^1(\frac{1}{2}\,|\dot{q}_i(t)-\dot{q}_j(t)|^2+\frac{N+3}{|q_i(t)-q_j(t)|} )dt  +\notag\\ &\qquad \qquad \qquad\sum_{\stackrel{1\leq i\leq N}{1\leq j\leq 3}}\int_0^1(\dfrac{1}{2}|\dot{q}_i(t)-\dot{q}_{N+j}(t)|^2
 \dfrac{N+3}{|q_i(t)-q_{N+j}(t)|})dt+\notag\\
& \qquad \qquad \qquad\sum_{1\leq i<j\leq N} \int_0^1(\frac{1}{2}\,|\dot{q}_i(t)-\dot{q}_j(t)|^2+\frac{N+3}{|q_i(t)-q_j(t)|} )dt  )\notag\\
&\geq\dfrac{3}{2}\times(\frac{4\pi^2}{N+3})^{\frac{1}{3}} [\, 3\times Nr(\dfrac{1}{Nr})^{\frac{1}{3}}+3\times(\dfrac{1}{3})^{\frac{1}{3}}C_N^2 +3N\,]\notag\\
&\stackrel{\triangle}{=} D.\tag{3.49}
\end{align}

When $N$ is odd, let $\tilde{A}=\inf{\{A,\, C,\,D\}}$, then on the collision set,
 the action functional $f\geq \tilde{A}$.

When $N$ is even, let $\tilde{B}=\inf{\{A,\,B,\, C,\,D\}}$, then
 on the collision set, the action functional $f\geq \tilde{B}$.

(1)Take $N=4, d=3,r=7,k_1=3,k_2=-4$.

We choose the following function as the test function:

Let $a>0,\,\,b>0$, and
\begin{align*}
&q_{i}=a(\cos{(6\pi t+\dfrac{2\pi (i-1)}{4})},\sin{(6\pi t+\dfrac{2\pi (i-1)}{4})}\,),\,\,i=1,\ldots,4,\\
&q_{j}=b(\cos{(-8\pi t+\dfrac{2\pi (j-5)}{3})},\sin{(-8\pi t+\dfrac{2\pi (j-5)}{3})}),\,\,j=5,6,7.
\end{align*}

We choose $a=0.2300,\,\,b=0.0880$,then
\begin{align*}
A\approx 144.6215,\,B\approx138.9586,\,&C\approx 170.7479,\,\,D\approx139.2196, \,\,\tilde{B}=138.9586,\\
&f(q)\approx135.5123< \tilde {B}.
\end{align*}
This proves that the minimizer of $f(q)$ on the closure $\bar{\Lambda}_2$ is a non-collision solution of the seven-body problem.

(2)Take $N=5, d=3,r=8,k_1=3,k_2=-5$.

We choose the following function as the test function:

Let $a>0,\,\,b>0$, and
\begin{align*}
&q_{i}=a(\cos{(6\pi t+\dfrac{2\pi (i-1)}{5})},\,\sin{(6\pi t+\dfrac{2\pi (i-1)}{5})}\,),\,\,i=1,\ldots,5,\\
&q_{j}=b(\cos{(-10\pi t+\dfrac{2\pi (j-6)}{3})},\,\sin{(-10\pi t+\dfrac{2\pi (j-6)}{3})}),\,\,j=6,7,8.
\end{align*}

 We choose $a=0.2450,\,\,b=0.0760$, then

 \begin{align*}
A\approx193.5057,\,\,&C\approx 181.0305,\,\,D\approx228.7437 , \,\,\tilde{A}=181.0305,\\
&f(q)\approx175.2312< \tilde {A}.
\end{align*}

This proves that the minimizer of $f(q)$ on the closure $\bar{\Lambda}_2$ is a non-collision solution of the eight-body problem.

(3)Take $N=7, d=3,r=10,k_1=3,k_2=-7$.

We  choose the following function as the test function:

Let $a>0,\,b>0$, and

\begin{align*}
&q_{i}=a(\cos{(6\pi t+\dfrac{2\pi (i-1)}{7})},\sin{(6\pi t+\dfrac{2\pi (i-1)}{7})}\,),\,\,i=1,\ldots,7,\\
&q_{j}=b(\cos{(-14\pi t+\dfrac{2\pi (j-8)}{3})},\sin{(-14\pi t+\dfrac{2\pi (j-8)}{3})}),\,\,j=8,9,10.
\end{align*}
We choose  $a=0.2500,\,\,b=0.0640$, then
 \begin{align*}
A\approx305.0645 ,\,\,
&C\approx274.1354,\,\,D\approx 360.6557, \,\,\tilde{A}=274.1354,\\
&f(q)\approx266.6297< \tilde {A}.
\end{align*}

This proves that the minimizer of $f(q)$ on the closure $\bar{\Lambda}_2$ is a non-collision solution of the ten-body problem.

\end{document}